\documentclass[11pt]{article}

\usepackage{amsfonts}
\usepackage{amsmath}
\usepackage{amssymb}
\usepackage{amsthm}
\usepackage{times}
\usepackage{hyperref}

\usepackage{url}
\usepackage{fullpage}
\usepackage{graphicx}
\usepackage{epsfig}
\usepackage{algorithm}
\usepackage{algorithmic}

\newtheorem{theorem}{Theorem}
\newtheorem{lemma}[theorem]{Lemma}
\newtheorem{claim}[theorem]{Claim}

\renewcommand{\setminus}{-}
\renewcommand{\P}{\mathrm{P}}
\newcommand{\NP}{\mathrm{NP}}
\newcommand{\DTIME}{\mathrm{DTIME}}

\newcommand{\ZPP}{\mathrm{ZPP}}
\newcommand{\polylog}{\mathrm{polylog}}
\newcommand{\poly}{\mathrm{poly}}
\newcommand{\dem}{\mathcal{D}}
\renewcommand{\deg}{\mathsf{deg}}
\newcommand{\mindeg}{\delta}
\newcommand{\maxdeg}{\Delta}
\newcommand{\maxleft}{\maxdeg_{\mathrm{left}}}
\newcommand{\maxright}{\maxdeg_{\mathrm{right}}}
\newcommand{\avgdeg}{\maxdeg^{\mathrm{avg}}}
\newcommand{\avgleft}{\avgdeg_{\mathrm{left}}}
\newcommand{\avgright}{\avgdeg_{\mathrm{right}}}
\newcommand{\minleft}{\mindeg_{\mathrm{left}}}
\newcommand{\minright}{\mindeg_{\mathrm{right}}}
\newcommand{\dist}{\mathsf{dist}}
\newcommand{\req}{\mathsf{req}}
\newcommand{\fixsmall}{\sigma}
\renewcommand{\Pr}{{\bf Pr}}
\newcommand{\E}{{\bf E}}

\newcommand{\hatG}{\widehat{G}}
\newcommand{\hatV}{\widehat{V}}

\newcommand{\hatW}{\widehat{W}}
\newcommand{\hatE}{\widehat{E}}

\newcommand{\squishlist}{
 \begin{list}{$\bullet$}
  { \setlength{\itemsep}{0pt}
     \setlength{\parsep}{3pt}
     \setlength{\topsep}{3pt}
     \setlength{\partopsep}{0pt}
     \setlength{\leftmargin}{1.5em}
     \setlength{\labelwidth}{1em}
     \setlength{\labelsep}{0.5em} } }

\newcommand{\squishlisttwo}{
 \begin{list}{$\bullet$}
  { \setlength{\itemsep}{0pt}
     \setlength{\parsep}{0pt}
    \setlength{\topsep}{0pt}
    \setlength{\partopsep}{0pt}
    \setlength{\leftmargin}{2em}
    \setlength{\labelwidth}{1.5em}
    \setlength{\labelsep}{0.5em} } }

\newcommand{\squishend}{
  \end{list}  }

\newcommand{\squishnum}{
 \begin{enumerate}
  { \setlength{\itemsep}{0pt}
     \setlength{\parsep}{3pt}
     \setlength{\topsep}{3pt}
     \setlength{\partopsep}{0pt}
     \setlength{\leftmargin}{1.5em}
     \setlength{\labelwidth}{1em}
     \setlength{\labelsep}{0.5em} } }

\newcommand{\squishnumend}{\end{enumerate}}

\title{Parameters of Two-Prover-One-Round Game and The Hardness of Connectivity Problems} 

\author{
Bundit Laekhanukit\thanks{
	School of Computer Science,
	McGill University, Montreal, QC, Canada. 
	\hbox{E-mail}:~{\tt blaekh@cs.mcgill.ca}
        This work was supported by the Natural Sciences and
        Engineering Research Council of Canada (NSERC) grant
        no.~288334 and 429598, by European Research Council (ERC)
        Starting Grant no.~279352 and by Harold H Helm fellowship. 
}
}

\date{\today}

\begin{document}

\maketitle

\begin{abstract}

Optimizing parameters of Two-Prover-One-Round Game (2P1R) is an
important task in PCPs literature as it would imply a smaller PCP with
the same or stronger soundness.
While this is a basic question in PCPs community, the connection
between the parameters of PCPs and hardness of approximations is
sometime obscure to approximation algorithm community.  
In this paper, we investigate the connection between the parameters of 
2P1R and the hardness of approximating the class of so-called
connectivity problems, which includes as subclasses the survivable
network design and (multi)cut problems.
Based on recent development on 2P1R by Chan (ECCC 2011) and several 
techniques in PCPs literature, we improve hardness results of some
connectivity problems that are in the form $k^\fixsmall$, for some 
(very) small constant $\fixsmall>0$, to hardness results of the form  
$k^c$ for some explicit constant $c$, where $k$ is a connectivity
parameter.  
In addition, we show how to convert these hardness into hardness
results of the form $\dem^{c'}$, where $\dem$ is the number of demand 
pairs (or the number of terminals). 
Our results are as follows.


\squishnum
\item For the rooted $k$-connectivity problem, we have hardness of 
   \[
     \left\{ \begin{array}{ll}
     k^{1/2-\epsilon} &
        \mbox{on directed graphs.} \\
     k^{1/10-\epsilon} &
        \mbox{on undirected graphs.} \\
     \dem^{1/4-\epsilon} &
        \mbox{on both directed and undirected graphs.} \\
     \end{array} \right.
  \]
  This improves upon the best known hardness of $k^{\fixsmall}$
  by Cheriyan et~al. (SODA~2012).
\item For the vertex-connectivity survivable network design problem,
  we have hardness of 
   \[
     \left\{ \begin{array}{ll}
     k^{1/6-\epsilon} &
        \mbox{on undirected graphs} \\
     \dem^{1/4-\epsilon} &
        \mbox{on both directed and undirected graphs.} \\
     \end{array} \right.
  \]
  This improves upon the best known hardness of $\Omega(k^{\fixsmall})$
  by Chakraborty et~al. (STOC~2008).
\item For the vertex-connectivity $k$-route cut problem on undirected
  graphs, we have hardness of 
   \[
     \left\{ \begin{array}{l}
     k^{1/6-\epsilon} \\
     \dem^{1/4-\epsilon} \\
     \end{array} \right.
  \]
  This improves upon the best known hardness of $k^{\fixsmall}$
  by Chuzhoy et~al. (SODA~2012).
\squishnumend
\end{abstract}

\pagebreak

\setcounter{page}{1}

\section{Introduction
	\label{sec:intro}}

Optimizing parameters of Two-Prover-One-Round Game (2P1R) is an
important task in PCPs literature as it would imply a smaller PCP with
the same or stronger soundness, which will in turn tighten hardness
results for many optimization problems. 
While this is a basic question in PCPs community, the connection
between the parameters of PCPs and hardness of approximations is
sometime obscure to approximation algorithm community.  
In this paper, we investigate the connection between the parameters of 
2P1R and the hardness of approximating the class of so-called
connectivity problems, which includes as subclasses the survivable
network design and (multi)cut problems.

Similar to 2P1R, a connectivity problem comes with several parameters,
e.g., the number of vertices $n$, a connectivity parameter $k$, and
the number of demand pairs $\dem$.
As these parameters are independent of each other, approximation
algorithms for connectivity problems are usually designed by
exploiting properties of the parameters, which means that the
approximation ratios of the algorithms depend on these terms.
By way of illustration, let consider a concrete example of 
the {\em rooted $k$-connectivity} problem on undirected graphs. 
In this problem, we are given an undirected graph $G=(V,E)$, a root
vertex $r$ and a set of terminals $T$; the goal is to find a
minimum-cost subgraph that has $k$ openly (vertex) disjoint paths from
the root vertex $r$ to each terminal $t\in T$. 
For arbitrary $k$, the best known approximation ratio of this problem
is $O(k\log{k})$ by Nutov~\cite{Nutov09b}, and it was shown by 
Cheriyan, Laekhanukit, Naves and Vetta~\cite{CLNV12} that the
dependence on $k$ cannot be taken out because the problem does not
admit $o(k^\fixsmall)$-approximation, for some (very) small constant 
$\fixsmall>0$, unless $\P=\NP$.
%
However, when $k$ is larger than the number of demands (or terminals)
$\dem$, a trivial $\dem$-approximation algorithm does exist and yields
a better approximation ratio than the $O(k\log{k})$-approximation
algorithm. Moreover, the hardness result of Cheriyan et~al. only
holds when $k$ is much smaller than $\dem$.
Thus, the approximability of the rooted $k$-connectivity problem on
undirected graphs depends on two parameters: the connectivity $k$ and
the number of demands (terminals) $\dem$, e.g., $k$ is a constant
independent of $\dem$.
Thus, to prove tighter approximation hardness of connectivity
problems, we have to consider all the parameters involved. 

Here two parameters of connectivity problems that we are interested in
are the connectivity  parameter $k$, which is the main focus in this
paper, and the number of demand pairs $\dem$. 
We consider 2P1R in its combinatorial form -- 
the {\em label-cover} problem. 
In this problem, we are given a bipartite directed graph $G=(U,W,E)$,
the set of labels (a.k.a., {\em alphabets}) $L$, and constraints which
are functions on edges $\{\pi_e:e\in{E}\}$; the goal is to 
find an assignment of labels to each vertex that satisfies all the 
constraints. 
It is known that the hardness of the label~cover problem depends on
two parameters the {\em maximum degree} $\maxdeg(G)$ of $G$ and the 
{\em alphabet-size} $|L|$. 
Thus, our goal is to investigate relationships between $k$, $\dem$ and
$\maxdeg(G)$, $|L|$. 

First, we consider the connectivity parameter $k$.
The problems whose hardness depending on the parameter $k$ are 
the rooted $k$-connectivity problem in both directed and undirected
graphs, the {\em vertex-connectivity survivable network design}
problem and the {\em vertex-connectivity $k$-route cut} problem. 
These problems have hardness of the form $k^\fixsmall$, where
$\fixsmall$ is a small constant that has not been calculated. 
(See \cite{Nutov09b,CLNV12,CCK08,CMVZ12}).
The common source of hardness of these problems is the label~cover
problem (a.k.a., 2P1R) with parallel repetition.  
The $\fixsmall$ here involved with the constant loss in the exponent when 
boosting the hardness gap using parallel repetition.
Estimating the value $\fixsmall$ is not an easy task, and even if
we can calculate this value, the constant is very small~\cite{KS11}.
By studying the constructions of all these problems, 
we observe that the connectivity parameter $k$ depends on the
maximum degree and alphabet-size of a label cover instance.  
Thus, the simplest way in proving this hardness is to find an instance
of the label~cover problem whose maximum degree and alphabet-size are
small comparing to the inverse of its soundness.
Based on the recent developments in
2P1R~\cite{Chan12,MR10,DKR12,AKS11}, we construct a label~cover
instance that satisfies the desired properties.
To be precise, we take a label~cover instance of Chan~\cite{Chan12}
that has alphabet-size close to the inverse of its soundness.
(Also, see the prior result by Khot and Safra~\cite{KS11}.)
Then we apply several reductions -- the {\em right degree} reduction
by Moshkovitz and Raz~\cite{MR10} and the {\em random sparsification} 
technique by Austrin, Khot and Safra~\cite{AKS11}\footnote{
Indeed, we are first inspired by the result of Moshkovitz and
Raz~\cite{MR10} and the result of Dinitz, Kortsarz and
Raz~\cite{DKR12}. However, due to a technical issue, we require a
technique in \cite{AKS11}, which was suggested by Siu On Chan.  
} to reduce the maximum degree of the instance.
Hence, we have an instance with small degree and small alphabet-size,
and we thus obtain the explicit exponent in the hardness of all the
problems mentioned above.

Second, we consider the parameter $\dem$, the number of demand pairs. 
The problems that we are interested in are the rooted $k$-connectivity
problem on both directed and undirected graphs, the
vertex-connectivity survivable network design problem and
the vertex-connectivity $k$-route cut problem.
By diving into the construction of these problems, we observe that
some of the demand pairs are independent and thus can be merged.
So, we have to partition the constraints (edges) of a label-cover
instance so that they have no conflict after reducing to a
connectivity problem.
We observe that such partitioning can be done using 
{\em strong edge coloring}.
To be precise, the strong edge coloring is a coloring of edges of $G$
such that, for any two edges $e,f$ with the same color, $e$ and $f$
share no endpoint, and $G$ has no edge joining an endpoint of
$e$ to an endpoint of $f$.
For example, edges $\{a,b\}$ and $\{c,d\}$ can have the same color if 
$a,b,c,d$ are all distinct vertices, and $G$ has none of the edges
$\{a,c\}$,$\{a,d\}$,$\{b,c\}$ and $\{b,d\}$.
It is known that a graph with maximum degree $\maxdeg$ has a 
strong edge coloring with $O(\maxdeg^2)$ colors. 
Thus, we can reduce the number of demands to be close to $\maxdeg(G)$,
which is thus close to the inverse of its soundness. 
%

Lastly, we would like to remark that we consider our results to be a
survey paper that connects the parameters of 2P1R to the hardness of
connectivity problems.
All the techniques used in this paper are not new and have been used
many times in literature.
The right~degree reduction was introduced by Moshkovitz and Raz in
\cite{MR10} and has been used in \cite{DH09}. 
The random sampling technique was used in PCPs literature by Goldreich
and Sudan in \cite{GS06} and was recently used by Dinitz, Kortsarz and
Raz in \cite{DKR12} to prove the hardness of the 
{\em basic $k$-spanner} problem.
Also, it has been used to reduced the degree of an instance of the
independent set problem by Austrin, Khot and Safra in \cite{AKS11}. 
Indeed, our work is inspired by the result of Moshkovitz and
Raz~\cite{MR10} and the result of Dinitz~et~al.~\cite{DKR12}. 
The graph coloring technique has been used to obtain approximation
algorithms for the rooted $k$-connectivity problem in undirected
graphs \cite{CK08,ChK09,Nutov09b}. 
Here we show that such technique can be used to show the converse,
i.e., the hardness of approximation.
(Indeed, to best of our knowledge, the strong edge coloring has not
been used in the previous literature.) 

The connectivity problems considered in this paper are as belows. 

\paragraph{The Rooted $k$-Connectivity Problem.}  
In the {\em rooted $k$-connectivity} problem, we are given a directed 
or undirected graph $G=(V,E)$ on $n$ vertices with cost $c_e$ on each
edge $e\in E$, a root vertex $r$, a set of terminals 
$T\subseteq V\setminus\{r\}$ and a connectivity requirement $k$. The
goal is to find a minimum-cost subgraph $G'=(V,E')$ of $G$ such that 
$G'$ has $k$ {\em openly (vertex) disjoint } paths from $r$ to each
terminal $t\in T$.  
This problem has been studied intensively in
\cite{CCK08,CK08,ChK08,CK09,Nutov09b,Nutov09c,ChK09,CLNV12}. 
The rooted $k$-connectivity problem is a fundamental network design
problem with vertex-connectivity requirements, and it lies at the
bottom of the complexity hierarchy of the vertex-connectivity
problems. 
In particular, the undirected rooted $k$-connectivity problem was
shown to be a special case of the {\em subset $k$-connectivity}
problem~\cite{Laekhanukit11b} and is clearly a special case of the 
{\em vertex-connectivity survivable network design} problem.
It can be seen that the same relationships also apply for the case of
directed graphs. 

The rooted $k$-connectivity problem on both directed and undirected
graphs admits a trivial $|T|$-approximation algorithm, which can be
done by applying a minimum-cost $k$-flow algorithm $|T|$ times, one for
each terminal. 
Non-trivial approximation algorithms for the rooted $k$-connectivity
problem are known only for the undirected case, and the best known
approximation ratio is $O(k\log{k})$ by Nutov~\cite{Nutov09b};
however, the approximation ratio surpasses that of the trivial
algorithm only when $k>|T|$.
On the negative side, Cheriyan, Laekhanukit, Naves and Vetta
\cite{CLNV12} recently showed that the rooted $k$-connectivity problem
on both directed and undirected graphs are hard to approximate to
within a factor of $k^\fixsmall$ for some fixed $\fixsmall>0$ 
(the constants $\fixsmall$ are different in directed and undirected
cases).
However, the constants $\fixsmall$ obtained are small and have not
been explicitly calculated. 

For the case of directed graphs, we give improved hardness of
$k^{1/2-\epsilon}$ and $\dem^{1/4-\epsilon}$ for the
rooted $k$-connectivity problem, for any constant $\epsilon>0$. 
(In fact, the $k^{1/2-\epsilon}$-hardness of this problem can be
derived from combining the result in \cite{CLNV12} and \cite{DKR12}.)
For the case of undirected graphs, the hardness are
$k^{1/10-\epsilon}$ and $\dem^{1/4-\epsilon}$, 
for any constant $\epsilon>0$, and
this also gives the same bound for the hardness of the subset
$k$-connectivity problem.
(Note that the number of demand pairs is $\dem=|T|$ for the rooted
$k$-connectivity problem and $\dem=|T|^2$ for the subset
$k$-connectivity problem.) 

\paragraph{The Vertex-Connectivity Survivable Network Design Problem.}
The {\em vertex-connectivity survivable network design} (VC-SNDP)
problem is a generalization of the rooted $k$-connectivity problem.
In this problem, we are given a directed or undirected graph
$G=(V,E)$ on $n$ vertices with a cost $c_e$ on each edge $e$ and a
connectivity requirement $\req(s,t)$ for each pair of vertices
$s,t\in{V}$.  
A vertex $s$ is called a {\em terminal} if there is a vertex $t$ such
that $\req(s,t)>0$, i.e., $s$ is a terminal if it has a positive
connectivity requirement; the set of terminals is denoted by $T$. 
The only known non-trivial approximation algorithm for this
problem due to the work of Chuzhoy and Khanna~\cite{CK09} has an
approximation ratio of $O(k^3\log{|T|})$, and the best known hardness
is $k^{\fixsmall}$, for some (very) small constant $\fixsmall>0$, due
to Chakrabarty, Chuzhoy and Khanna~\cite{CCK08}.
We give an improved hardness of $k^{1/6-\epsilon}$ and 
$\dem^{1/4-\epsilon}$ for VC-SNDP, for any constant $\epsilon>0$.

\paragraph{The Vertex-Connectivity $k$-Route Cut Problem.}
In the {\em vertex-connectivity $k$-route cut} (VC-$k$-RC)
problem, we are given an undirected graph $G=(V,E)$ on $n$ vertices
with a cost $c_e$ on each edge $e\in E$, a set of source-sink pairs
$\{(s_1,t_1),(s_2,t_2),\ldots,(s_{\dem},t_{\dem})\}\subseteq{V\times{V}}$  
and a connectivity parameter $k$.
The goal is to find a minimum-cost subset $E'\subseteq E$ of edges
such that $G\setminus E'$ has no $k$ openly disjoint $s_i,t_i$-paths
for every source-sink pairs $s_i,t_i$. 
The best known approximation guarantee for this problem is
$O(\dem\cdot{k})$ due to the work of Chuzhoy, Makarychev,
Vijayaraghavan and Zhou~\cite{CMVZ12}, and 
the best known hardness is $O(k^\fixsmall)$, for some (very) small
constant $\fixsmall>0$.
The approximation ratio is slightly better when we turn to a
bi-criteria approximation algorithm.
Chuzhoy~et~al. showed that there is an algorithm that guarantees to
find a solution $E'\subseteq{E}$ with cost at most
$O(\lambda{k}\log^{2.5}\dem\log\log\dem)$ times the optimal, 
where $\lambda$ is the maximum number of demand pairs in
which any terminal participates, and $E'$ 
{\em cuts at least $k/2$-routes}, i.e., $G\setminus{E'}$   
has no $k/2$ openly $s_i,t_i$-paths for all $i$.  
In this paper, we show that at least one of the two terms $k$ and
$\dem$ cannot be taken out. Precisely, we show that it is hard to
approximate VC-$k$-RC to within a factor of  
$k^{1/6-\epsilon}$ and $\dem^{1/4-\epsilon}$, for any constant
$\epsilon>0$. 

Our hardness results are summarized in Table~\ref{tab:results}.

\begin{table} 
\begin{center}
\begin{tabular}{|c|c|c|c|}
\hline
Problem & Graphs & In terms of $k$ ($k < \dem$) 
                 & In terms of $\dem$ ($k \geq \dem$) \\
\hline
Rooted $k$-Connectivity & Directed &
 $k^{1/2-\epsilon}$
   & $\dem^{1/4-\epsilon}$ \\
 & Undirected &
 $k^{1/10-\epsilon}$ & $\dem^{1/4-\epsilon}$ \\
\hline
Subset $k$-Connectivity & Undirected & 
 $k^{1/10-\epsilon}$ & $\dem^{1/4-\epsilon}$ \\
\hline
VC-SNDP & 
 Undirected &
 $k^{1/6-\epsilon}$ & $\dem^{1/4-\epsilon}$ \\ 
\hline
VC-$k$-Route Cut & 
 Undirected &
 $k^{1/6-\epsilon}$ & $\dem^{1/4-\epsilon}$ \\ 
\hline
\end{tabular}
\end{center}
\caption{
The table summarizes our hardness results, which hold for any
$\epsilon>0$.} 
\label{tab:results}
\end{table}



\section{Preliminaries}
\label{sec:prelim}

We use standard graph terminologies as in \cite{DiestelBook}.  
Let $G=(V,E)$ be any graph.
For any vertex $v\in V$, the degree of $v$ in $G$ is denoted by
$\deg_G(v)$.
The maximum (resp., minimum) degree of $G$, denoted by $\maxdeg(G)$ 
(resp., $\mindeg(G)$), is the maximum (resp., minimum) degree over all
vertices of $G$.
If we consider more than one graph, then we denote the set of vertices
and edges of $G$ by $V(G)$ and $E(G)$, respectively.

By a bipartite directed graph, we mean a directed graph $G=(U,W,E)$
such that every arc is directed from $U$ to $W$, i.e., an arc of $G$
is of the form $(u,w)$, where $u\in U$ and $w\in W$. 
We call vertices in $U$ {\em left vertices} and vertices in $W$ 
{\em right vertices}. 
Since each left (resp., right) vertex of $G$ has no incoming 
(resp., outgoing) arc, we abuse the term ``degree'' to mean indegree
(resp., outdegree) of left (resp., right) vertices of $G$.
By the {\em maximum (resp., minimum) left degree} of $G$, denoted by 
$\maxleft(G)$ (resp., $\minleft(G)$), we mean the maximum (resp.,
minimum) degree of left vertices of $G$.
Similar, notations are used for right vertices.
Thus, $\maxright(G)$ (resp., $\minright(G)$) denotes the maximum  
(resp., minimum) degree of right vertices of $G$.
We use a similar notations for the average degree of $G$.
The average degree of $G$ is denoted by $\avgdeg(G)$, and the average 
left and right degree of $G$ are denoted by $\avgleft(G)$ and
$\avgright(G)$, respectively.  
We say that $G$ is {\em left (resp., right) regular} if every left 
(resp., right) vertex of $G$ has the same degree. 
If $G$ is both left and right regular with degree $d_1$ and $d_2$,
then we say that $G$ is {\em $(d_1,d_2)$-regular}.
If $G$ is clear in the context, then we will omit $G$, e.g., we may
write $\deg_G(v)$ as $\deg(v)$ and write $\maxdeg(G)$ as $\maxdeg$.

By a {\em matching} $M$ of a (directed) graph $G$, we mean a set
of edges (resp., arcs) such that no two edges (resp., arcs) in $M$
share an endpoint, and by {\em induced matching} $I$ in $G$, we mean a
matching such that no edge (resp., arc) in $G$ joins endpoints of
edges (resp., arcs) in $I$. 
Thus, a subgraph of $G$ induced by such $I$ is also a matching. 
A {\em strong edge coloring} of $G$ is a partition
$E_1,E_2,\ldots,E_\ell$ of sets of edges (resp., arcs) of $G$ such
that each $E_i$ is an induced matching in $G$. 
The smallest number $\ell$ such that $G$ has an $\ell$-strong edge
coloring is called the {\em strong chromatic index} of $G$, denoted by
$\chi'_S(G)$.

All of our hardness results come from the same source, 
{\em the label~cover} problem (a.k.a, 2P1R).
Hence, we devote the next section to discuss the label~cover problem.   

\subsection{The Label~Cover Problem}
\label{sec:label-cover}

The {\em (maximum) label~cover} problem ({\em the projection game}) is
defined as follows. 
We are given a directed bipartite graph $G=(U, W, E)$ on $n$ vertices,  
two sets of labels (a.k.a, alphabets) $L_1$ (for vertices in $U$) and 
$L_2$ (for vertices in $W$), and a {\em constraint} $\pi_e$ on each
arc $e$, which is a {\em projection}\footnote{
The constraints of the label~cover problem can be relations instead of 
projections; however, here we define the label~cover problem as the
projection game.
}
$\pi_e:L_1\rightarrow L_2$.
A {\em labeling} $(f_1,f_2)$ is a pair of functions $f_1:U\rightarrow
L_1$ and $f_2:W\rightarrow L_2$ assigning a label to each vertex of
$U$ and $W$, respectively.
We say that $(f_1,f_2)$ {\em covers} an arc $(u,w)\in E$ if
$\pi_e((f_1(u))=f_1(w)$. 
The goal in the maximum label~cover problem is to find a labeling that
maximizes the number of arcs covered.
For notational convenience, we shall denote an instance of the
label~cover problem by $(G=(U,W,E),\{\pi_e:e\in{E}\},L_1,L_2)$.

The gap version of the maximum label~cover problem is the problem of
deciding whether a given instance of the maximum label~cover problem
is one of the following two cases:
\squishlist
\item {\sc Yes-Instance:} There is an labeling covering at least
  $(1-\epsilon)$ fraction of all the arcs.
\item {\sc No-Instance:} There is no labeling covering more than
  $\gamma$ fraction of all the arcs.
\squishend
We call $1-\epsilon$ and $\gamma$ the {\em completeness} and the 
{\em soundness} of the label~cover instance, respectively.
If $\epsilon=0$, then we say that a (gap) label~cover instance has
{\em perfect completeness}; otherwise, we say that an instance has
{\em imperfect completeness}. 
It can be seen that NP-hardness of the gap version of the label~cover
problem implies the hardness of the maximum one. 
Thus, we shall abuse the term ``maximum label~cover'' to also mean
the gap label~cover problem. 

For our purpose, we need a minimization version of the label~cover
problem, which can be defined by allowing each vertex to have more
than one label, and the goal is to minimize the total cost of
labels used over all vertices.  
To be precise, we define the {\em minimum-cost label~cover}
problem to be the weighted counter part of the maximum label~cover 
problem. 
The minimum label~cover problem was defined in \cite{ABSS97}, and it  
has an equivalent form known as the {\em Min-Rep} problem as defined
in \cite{Kortsarz01}. 
The input of this problem is the same as that of the maximum
label~cover problem except that we also have a cost $c_1$ on each
label $a\in L_1$ and a cost $c_2$ on each label $b\in L_2$. 
The labeling is relaxed as a pair of functions
$(f_1,f_2)$, where $f_1:U\rightarrow 2^{L_1}$ and 
$f_2:W\rightarrow 2^{L_2}$, i.e., we are allowed to assign more than
one labels to each vertex.
A labeling $(f_1,f_2)$ {\em covers} an arc $e=(u,w)$ if there are
labels $a\in f_1(u)$ and $b\in f_2(w)$ such that $\pi_e(a)=b$. 
The goal in the minimum-cost label~cover problem is to find a
labeling $(f_1,f_2)$ that covers all the arcs and minimizes the cost 
$c(f_1,f_2)=\sum_{u\in{U}}c_1\cdot|f_1(u)|+\sum_{w\in{W}}c_2\cdot|f_2(w)|$. 

Note that there is a standard technique that transforms the hardness
of the maximization version of the label~cover problem to the 
minimum-cost version. 
(See Appendix~\ref{sec:max-to-min} for more detail.)
Thus, it suffices to consider the maximum label~cover problem.

\section{Relationships between Label~Cover and with Connectivity Problems}
\label{sec:apps}

Here we show the relationships between the parameters of the
label~cover problem (2P1R) with the hardness of approximating
connectivity problems.

First, we survey relationships between the hardness in terms of the
connectivity parameter $k$ of the connectivity problems and the
parameters of the minimum-cost label~cover problems.

\begin{theorem}[\cite{CLNV12,CCK08,CMVZ12}]
\label{thm:lc-to-conn}
Given an instance
$(G=(U,W,E),\{\pi_e:e\in{E}\},L_1,L_2,c_1,c_2)$ of the minimum-cost
label~cover problem, there are polynomial-time approximation
preserving reductions that output
\squishlist
\item An instance of the rooted $k$-connectivity problem on directed
  graphs with $k=\maxdeg(G)$. 
\item An instance of the rooted $k$-connectivity problem on undirected
  graphs with
  $k=O(\maxdeg(G)^3\cdot\max\{|L_1|,|L_2|\} + \maxdeg(G)^4)$.
\item An instance of the vertex-connectivity survivable network design
  problem on undirected graphs with a maximum requirement
  $k=O(\maxdeg(G)\cdot\max\{|L_1|,|L_2|\} + \maxdeg(G)^2)$. 
\item An instance of the vertex-connectivity $k$-route cut problem on
  undirected graphs with $k=O(\maxdeg(G)\cdot\max\{|L_1|,|L_2|\})$.
\squishend
\end{theorem}

The hardness in terms of the connectivity parameter $k$ can be
transformed into hardness in terms of the number of demand pairs
$\dem$. 
The parameter that involves with hardness in this term is the degree
of the label~cover instance. 
We claim that, for each of the problems we consider, two demand pairs
$(s_1,t_1)$ and $(s_2,t_2)$ are {\em independent} if and only if they
come from two different constraints (arcs) $(u_1,w_1)$ and $(u_2,w_2)$
of the label~cover instance such that $(u_1,w_1)$ and $(u_2,w_2)$
forms an induced matching, which thus can have 
the same ``strong edge color''.
So, we can partition the arcs of the label~cover instance using strong
edge coloring and merge source-sink pairs with the same color.
Thus, we have the following theorem.


\begin{theorem} \label{thm:lc-to-demands}
For each of the following problems, say $\Pi$, 
\squishlist
\item The rooted $k$-connectivity problem on directed graphs,
\item The rooted $k$-connectivity problem on undirected graphs,
\item The vertex-connectivity survivable network design problem on
  undirected graphs,
\item The vertex-connectivity $k$-route cut problem on undirected
  graphs,
\squishend
there is a polynomial-time reduction that, given 
an instance $(G=(U,W,E),\{\pi_e:e\in{E}\},L_1,L_2,c_1,c_2)$ of the
minimum-cost label~cover problem, outputs an instance of the problem
$\Pi$ with the number of demand pairs $\dem=2\maxdeg(G)^2$. 
\end{theorem}

See Appendix~\ref{sec:dir-rooted}, \ref{sec:undir-rooted},
\ref{sec:vc-sndp}, and \ref{sec:k-route-cut} for the
full proofs and discussions.

As we will show in the next section, the label~cover instance of
Chan~\cite{Chan12} can be modified so that it has degree close to
the inverse of its soundness. 
(See Theorem~\ref{thm:small-strong-lc}). 
We apply a standard technique to transform the hardness of the
maximization version of the label~cover problem to the 
minimum-cost version. 
(See Appendix~\ref{sec:max-to-min} for more detail.)
Then we have the following theorem.  

\begin{theorem} \label{thm:mincost-small-strong-lc}
For any constants $q>0$ and $\epsilon>0$, given an instance  
$(G=(U,W,E),\{\pi_e:e\in E\},L_1,L_2)$ of the maximum label~cover
such that $|L_1|,|L_2|\leq q^2$, $\maxdeg(G) = \Theta(q)$ and 
$\avgdeg(G) = \Theta(q)$, unless $\NP=\ZPP$, it is hard to approximate
the minimum-cost label~cover problem to within a factor of
$q^{1/2-\epsilon}$. 
\end{theorem}

By substituting the bound in Theorem~\ref{thm:mincost-small-strong-lc}
to Theorem~\ref{thm:lc-to-conn} and Theorem~\ref{thm:lc-to-demands},
we prove the results in Table~\ref{tab:results}. 



\section{Modifying The Label~Cover Instance}
\label{sec:transform-construct-reg-lc}

In this section, we show how to construct a label~cover instance with 
strong soundness, small degree and small alphabet-size.
In particular, we prove the following lemma.

\begin{lemma} \label{lmm:degree-reduction}
Let $q>0$ be a constant.
There is a randomized polynomial-time algorithm that reads as input 
an instance $(G=(U,W,E),\{\pi_e:e\in E\},L_1,L_2)$ of the maximum
label-cover problem with the following properties:
\squishlist
\item The alphabet-size is $\max\{|L_1|,|L_2|\}\leq q$.
\item The graph $G$ has regular left degree $\maxleft=\poly(q)$. 
\item The completeness is $1-\epsilon$, for any $\epsilon>0$. 
\item The soundness is $\gamma=1/\poly(q)$.
\squishend
outputs an instance $(G'=(U',W',E'),\{\pi_e:e\in E\},L_1,L_2)$ of the
maximum label-cover problem with completeness $1-\epsilon$ and
soundness $\gamma'=\Theta(\gamma)$ and
$\maxdeg(G')\leq{O((1/\gamma)\log(1/\gamma))}$.
\end{lemma}

The following theorem is due to the work of Chan~\cite{Chan12}.

\begin{theorem}[\cite{Chan12}] \label{thm:strong-lc}
For any constants $q>0$ and $\epsilon>0$, given an instance  
$(G=(U,W,E),\{\pi_e:e\in E\},L_1,L_2)$ of the maximum label~cover
such that $|L_1|,|L_2|\leq q^2$ and $\maxleft(G)=q$, it is NP-hard to
distinguish between the following  
two cases.
\squishlist
\item {\sc Yes-Instance:} There is an labeling covering at least
  $(1-\epsilon)$ fraction of all the arcs.
\item {\sc No-Instance:} There is no labeling covering more than
  $O(\log{q}/q)$ fraction of all the arcs.
\squishend
\end{theorem}

Thus, by invoking Lemma~\ref{lmm:degree-reduction}, we have the
following theorem.
 
\begin{theorem} \label{thm:small-strong-lc}
For any constants $q>0$ and $\epsilon>0$, given an instance  
$(G=(U,W,E),\{\pi_e:e\in E\},L_1,L_2)$ of the maximum label~cover
such that $|L_1|,|L_2|\leq q^2$, $\maxdeg(G) = \Theta(q\log{q})$ and  
$\avgdeg(G) = \Theta(q\log{q})$, unless $\NP=\ZPP$, it is hard to
distinguish between the following two cases. 
\squishlist
\item {\sc Yes-Instance:} There is an labeling covering at least
  $(1-\epsilon)$ fraction of all the arcs.
\item {\sc No-Instance:} There is no labeling covering more than
  $O(\log{q}/q)$ fraction of all the arcs.
\squishend
\end{theorem}

So, we devote the remaining part of this section to prove
Lemma~\ref{lmm:degree-reduction}. 
We have four steps. 
First, we take a basic instance, which is a label~cover instance with
strong soundness and have a regular left-degree as in
Theorem~\ref{thm:strong-lc}. 
We apply the right~degree reduction to make a $(d_1,d_2)$-regular
instance. 
Then we make copies of left vertices so that both sides have the
same number of vertices and thus have regular degree.
Finally, we apply a random sparsification to reduce the maximum degree
of a label~cover graph to be $O(q\log{q})$, where
$\gamma=1/\poly(q)$ is the soundness of the label~cover instance.

\subsection{Basic Instance}
\label{sec:transform:basic}

We take an instance $(G=(U,W,E),\{\pi_e:e\in E\},L_1,L_2)$ of the
maximum label~cover problem with properties as stated in
Lemma~\ref{lmm:degree-reduction}. That is, 
\squishlist
\item The alphabet-size is $\max\{|L_1|,|L_2|\}\leq q$.
\item The graph $G$ has regular left degree $D=\poly(q)$. 
\item The completeness is $1-\epsilon$, for any $\epsilon>0$. 
\item The soundness is $\gamma=1/\poly(q)$.
\squishend

An instance of the maximum~label cover problem that satisfies the
above properties are that constructed by Chan in \cite{Chan12} and by
Khot and Safra in \cite{KS11}.
Note that due to the size of the construction, the former result
applies for any constant $q>0$ while the latter result in \cite{KS11}
applies for all primes  $5 \leq q \leq \polylog(N)$, where $N$ is the
size of the label~cover instance. 
More precisely, the result in \cite{KS11} also applies for
$q=\polylog(N)$ under the hardness assumption
$\NP\subsetneq\DTIME(2^{\polylog{n}})$.

\subsection{Making An Instance $(d_1,d_2)$-Regular}
\label{sec:transform:right-deg}

The basic instance discussed in the previous chapter is the bipartite
graph $G=(U,V,E)$ that is left-regular but not right-regular.
To make an instance of the maximum label~cover instance regular, we
apply the {\em right degree reduction} introduced by Moshkovitz and
Raz~\cite{MR10}. (Also, see~\cite{DH09}.)
In short, the right degree reduction makes the right degree of a
label~cover instance regular while almost preserves the soundness.
It is not hard to see that the reduction preserves the completeness as
well.
(See Appendix~\ref{sec:right-degree-reduction} for more detail.)

\begin{lemma}[Right Degree Reduction~\cite{MR10}]
\label{lmm:deg-reduction}
There exists a polynomial-time reduction that, given 
a parameter $d$ and a maximum label~cover instance 
$(G=(U,W,E),L_1,L_2,\{\pi_e:e\in{E}\})$ with completeness
$1-\epsilon$ and soundness $\gamma$, where $G$ has regular left
degree $\maxleft$, outputs a maximum label~cover instance
$(G'=(U',W',E'),L_1,L_2,\{\pi_e\}_{e\in E'})$ with 
regular left degree $d\cdot\maxleft$, regular right degree $d$, 
completeness $1-\epsilon$ and soundness $\gamma+O(1/\sqrt{d})$. 
\end{lemma}

We choose a parameter $d=1/\gamma$ and apply the right~degree
reduction on $G$.
Thus, we have an instance
$(\hatG=(U^{k},\hatW,\hatE),L_1,L_2,\{\pi_e\}_{e\in\hatE})$ 
of the maximum label~cover problem in which $\hatG$ is
$(dD,d)$-regular, where $D$ is the left-degree of $G$, as desired. 

\subsection{Making $(d_1,d_2)$-Regular Instance $\maxdeg$-Regular} 
\label{sec:transform:make-regular}

Take an instance $\hatG$ of the maximum label~cover problem as
discussed in the previous section.
Now, we want to make the $(dD,d)$-regular graph $\hatG$ a
$dD$-regular graph.  
To do so, we replace each left vertex $u$ of $G$ by $D$ vertices
$u_1,u_2,\ldots,u_D$ and we add an arc $(u_i,w)$ with a constraint 
$\pi_{u_i,w}=\pi_{u,w}$, for each arc $(u,w)$ of $\hatG$.   
This results in a graph $G^{reg}$ which is $dD$-regular
because the degree of each right vertex increases by a factor of $D$
while the degree of each left vertex remains the same. 
Observe that the reduction preserves completeness because each edge of
$\hatG$ has exactly $D$ copies in $G^{reg}$.
Now, consider the soundness.
Take any labeling $(f_1,f_2)$ of $G^{reg}$.
We construct a labeling $(\hat{f}_1,f_2)$ of $\hatG$ by assigning
$\hat{f}_1(u)=f_1(u_i)$, where $f_1(u_i)$ is a labeling that covers
the maximum number of arcs of $G^{reg}$ incident to $u_i$ given that
$f_2$ is fixed. 
If $(f_1,f_2)$ covers more than $\gamma$ fraction of arcs of
$G^{reg}$, then $(\hat{f}_1,f_2)$ will cover more than $\gamma$
fraction of arcs of $\hatG$ as well by the choice of $\hat{f}_1(u)$. 
Therefore, the reduction preserves both completeness and soundness,
and the resulting bipartite graph is $\maxdeg$-regular, where
$\maxdeg=dD$.

\subsection{Reducing Degree via Random Sparsification}
\label{sec:rand-sparse}

Now, we take a $dD$-regular label~cover instance from the previous
section. 
We apply a random sparsification technique to reduce the 
``average degree'' of the label~cover instance to almost match the
inverse of its soundness. 
Then we throw away vertices with large degree so that the graph has
degree within the desired bound, the inverse of the soundness.

\subsubsection{Sparsifying The Graph}
\label{sec:sparse}

First, we will sparsify the graph $\hatG$.
The reduction takes as input a regular-degree instance of the maximum
label~cover problem and outputs an instance whose bipartite graph has
small average degree. 
To be precise, the input of our reduction is an instance
$(G=(U,W,E),\{\pi_e:e\in E\},L_1,L_2)$ of the maximum label~cover
problem with regular degree $\maxdeg$, completeness $1-\epsilon$ and
soundness $\gamma$.
Then it constructs a graph $G'=(U,W,E')$ from $G$ by randomly and
independently picking each arc of $G$ with probability
$\rho=\gamma^{-1}\log(\max\{|L_1|,|L_2|\})/\Delta(G)$.

Intuitively, since we sample arcs of $G$ with the same probability
$\rho$, the resulting graph $G'$ should have degree approximately
$O(\rho\maxdeg)$, and for any labeling $(f_1,f_2)$, the fraction of
arcs in $G'$ that $(f_1,f_2)$ covers is approximately the same as that
it covers in $G$.  
The next theorem shows that the random sparsification (almost)
preserves completeness and soundness of the original instance. 
Moreover, the average degree of the output instance is exactly
$\avgdeg(G)=\gamma^{-1}\log(\max\{|L_1|,|L_2|\})$.

\begin{lemma}\label{lmm:sparse}
Suppose the random sparsification algorithm takes as input an instance
$(G=(U,W,E),\{\pi_e:e\in E\},L_1,L_2)$ of the maximum
label~cover problem with regular degree $\maxdeg(G)$, completeness
$1-\epsilon$ and soundness $\gamma$, where
$0<\epsilon,\gamma<1$. 
Then it outputs with high probability an instance
$G'=(U,W,E',\{\pi_e:e\in E'\},L_1,L_2)$ of the maximum label~cover
problem with completeness $1-4\epsilon$, soundness $8\gamma$ and  
the average degree of $G'$ is 
$\avgdeg(G')=\gamma^{-1}\log(\max\{|L_1|,|L_2|\})$.
\end{lemma}

\begin{proof}

Throughout, let $n=|U|+|W|$ denote the number of vertices of $G$.  

\medskip

{\bf Completeness: } Suppose there is a labeling $(f_1,f_2)$
covering $(1-\epsilon)$ fraction of arcs in $G$.
We will show that $(f_1,f_2)$ covers at least $1-2\epsilon$
fraction of arcs in $G'$.

Let $X=\sum_{e\in E}X_e$ be the number of arcs covered by the labeling
$(f_1,f_2)$, where $X_e$ is an indicator random variable such that
$X_e=1$ if an arc $e$ is covered by the labeling $(f_1,f_2)$ and 
$X_e=0$ otherwise.
Then the expected number of arcs not covered by $(f_1,f_2)$ is 
\[
\E[|E|-X] 
= \epsilon |E|
  \cdot \frac{\gamma^{-1}}{\maxdeg(G)}\log(\max\{|L_1|,|L_2|\}) 
= \frac{\epsilon n}{2}\cdot\gamma^{-1}\log(\max\{|L_1|,|L_2|\}) 
\]

The last equation follows since $|E|=\maxdeg(G)|U|=\maxdeg(G)|W|$. 
By Chernoff's bound, we have

\[
\Pr[|E|-X > \epsilon{n}\gamma^{-1}\log(\max\{|L_1|,|L_2|\})] 
< \exp\left
  (-\frac{\epsilon n}{6}\cdot\gamma^{-1}\log(\max\{|L_1|,|L_2|\}) 
\right) 
\leq 2^{-\Omega(n)}
\]

Now, consider the expected number of arcs in $G'$.
We have
\[
\E[|E'|] =
|E|\cdot\frac{\gamma^{-1}\log(\max\{|L_1|,|L_2|\})}{\maxdeg(G)}
= \frac{n}{2}\cdot\gamma^{-1}\log(\max\{|L_1|,|L_2|\}).
\]

Thus, by Chernoff's bound, 
\[
\Pr\left[|E'|<\frac{n}{4}\gamma^{-1}\log(\max\{|L_1|,|L_2|\})\right] 
< \exp\left(-\frac{n}{16}\gamma^{-1}\log(\max\{|L_1|,|L_2|\})\right)
\leq 2^{-\Omega(n)}
\]

By union bound, with high probability, $(f_1,f_2)$ covers at 
least $(1-4\epsilon)$ fraction of arcs in $G'$.

\medskip

{\bf Soundness: }  Suppose there is no labeling $(f_1,f_2)$
covering more than $\gamma$ fraction of arcs in $G$.
We will show that there is also no labeling $(f_1,f_2)$ covering
more than $4\gamma$ fraction of arcs in $G'$.

Fix any labeling $(f_1,f_2)$. 
Let $X=\sum_{e\in E}X_e$ be the number of arcs covered by the labeling
$(f_1,f_2)$, where $X_e$ is an indicator random variable such that
$X_e=1$ if an arc $e$ is covered by the labeling $(f_1,f_2)$ and 
$X_e=0$ otherwise.
Then the expected number of arcs satisfied by $(f_1,f_2)$ is 
\[
\E[X] =
\gamma|E|\cdot\frac{\gamma^{-1}\log(\max\{|L_1|,|L_2|\})}{\maxdeg(G)}
= \frac{n}{2}\cdot\log(\max\{|L_1|,|L_2|\}).
\]

Thus, by Chernoff's bound, we have 
\[
\Pr[X > 2n\log(\max\{|L_1|,|L_2|\})]
<
\exp\left(\frac{9}{6}n\log(\max\{|L_1|,|L_2|\})\right) 
\leq \left(\frac{1}{\max\{|L_1|,|L_2|\}}\right)^{3n/2}
\]

Since there are $|L_1|^{n/2}|L_2|^{n/2}\leq
(\max\{|L_1|,|L_2|\})^n$ possible labellings, by union bound, we have
that with probability $2^{-\Omega(n)}$, the labeling $(f_1,f_2)$ covers
at most $2n\log(\max\{|L_1|,|L_2|\})$ arcs.
By the proof for the case of {\bf Yes-Instance}, we
have that $G'$ has at most $(n/4)\gamma^{-1}\log(\max\{|L_1|,|L_2|\})$ 
arcs with probability $2^{-\Omega(n)}$. 
Thus, the labeling $(f_1,f_2)$ covers at most $8\gamma$ fraction of
the arcs with high probability. This completes the proof.
\end{proof}

So, we can sparsify the instance
$(\hatG=(U,W,\hatE),\{\pi_e:e\in{\hatE}\},L_1,L_2)$ 
from Section~\ref{sec:transform:make-regular} to obtain an instance 
$(\hatG^{avg}=(U,W,\hatE^{avg}),\{\pi_e:e\in{\hatE^{avg}}\},L_1,L_2)$
such that $\hatG^{avg}$ has average degree
$\avgdeg(G')=\gamma^{-1}\log(\max\{|L_1|,|L_2|\})$.

\subsubsection{Removing Vertices with Large Degree}

The graph $\hatG^{avg}$ obtained from the previous step has average
degree to be within the desired bound.
However, some vertices may still have {\em large degree}, i.e., their
degree are larger than $2 \gamma^{-1}\log(\max\{|L_1|,|L_2|\})$.  
To make the graph to have degree within the desired bound, we remove
all the large degree vertices from the graph $\hatG^{avg}$.
This results in a graph $\hatG^{bound}=(U',V',E')$ with 
$\maxdeg(\hatG^{bound})\leq 2 \gamma^{-1}\log(\max\{|L_1|,|L_2|\})$. 

By Chernoff's bound, the probability that a vertex $v\in{U\cup{W}}$ has
large degree is 
\[
\Pr[\deg_{\hatG^{avg}}(v) > 2 \gamma^{-1}\log(\max\{|L_1|,|L_2|\})]
 < \exp\left(- \frac{1}{3}\gamma^{-1}\log(\max\{|L_1|,|L_2|\}\right)
 \leq 2^{-(1/3)\gamma^{-1}}
\]

Let $X_v$ be an indicator variable such that $X_v=1$ if $v$ has large
degree and $x_v=0$ otherwise. Then we have 
$
  \E\left[\sum_{v\in{V\cup{W}}}X_v\right]
  = \sum_{v\in{V\cup{W}}}\E[X_v]
  \leq 2^{-(1/3)\gamma^{-1}}|V\cup{W}|$.

By Markov's inequality, we have
$\Pr\left[\sum_{v\in{V\cup{W}}}X_v >
   2^{-(1/3)\gamma^{-1}+1}|V\cup{W}|\right] < \frac{1}{2}$.


Thus, with probability $1/2$ we
remove at most $2^{-(1/3)\gamma^{-1}}|V\cup{W}|$ vertices of
$\hatG^{avg}$; we call this a {\em probability of success}.
We can repeat the process $O(\log{n})$ times, where $n=|U|+|W|$, to 
increase the probability of success to $1-1/\Omega(n)$.
This does not effect the success probability of the random
sparsification step because the probability of success of the random
sparsification step is very high, say $1-1/2^{\Omega(n)}$.


\medskip

\noindent{\bf The Size of Construction:}
As above, with high probability, the graph $\hatG^{bound}$ has at
least  $(1-1/2^{n/3})n$ vertices.
For the number of arcs, we may assume that all vertices removed have
degree $dD=\poly(\gamma^{-1})$, where $dD$ is the (regular) degree of
the graph $\hatG$. 
Note that $\gamma^{-1}$ is smaller than $O(\log{n})$. 
So, with high probability, the number of arcs of $\hatG^{bound}$ is at
least 
$|E(\hatG^{avg})| - 2^{-n/3}dD(|U|+|W|) 
 \geq (1-2^{-n/6})|E(\hatG^{avg})|$. 

\medskip 

\noindent{\bf Completeness:}
Suppose there is a labeling $(f_1,f_2)$ covering 
$(1-\epsilon)$ fraction of the arcs of $\hatG^{avg}$.
We will show that $(f_1,f_2)$ covers at least 
$(1-2\epsilon)$ fraction of the arcs of $\hatG^{bound}$. 
 
We may assume that arcs incident to vertices removed from
$\hatG^{avg}$ are covered by $(f_1,f_2)$, and each vertex removed has
degree $dD=\poly(\gamma^{-1})$, which is the maximum degree of
$\hatG$. 
Thus, the number of arcs of $\hatG^{bound}$ covered by $(f_1,f_2)$ is
at least
$ (1-\epsilon)|E(\hatG^{avg})| - 2^{-n/3}dD(|U|+|W|)
\geq 
(1 - 2\epsilon)|E(\hatG^{avg})|$.
%

The last inequality follows because $\epsilon$ is a constant. 
Therefore, in this case, there is a labeling that covers at least
$(1-2\epsilon)$ fraction of the arcs of $\hatG^{bound}$. 

\medskip

\noindent{\bf Soundness:}
Suppose there is no labeling $(f_1,f_2)$ that covers more than
$\gamma$ fraction of the arcs of $\hatG^{avg}$. 
We will show that there is no labeling $(f_1,f_2)$ that covers more
than $2\gamma$ fraction of the arcs of $\hatG^{bound}$.  

Consider a labeling $(f_1,f_2)$ of $\hatG^{avg}$ that covers
$\gamma'\leq\gamma$ fraction of arcs of $\hatG^{bound}$. 
We construct a labeling $(f'_1,f'_2)$ of $\hatG^{bound}$ by assigning
$f'_1(u)=f_1(u)$ (resp., $f'_2(w)=f_2(w)$) for each vertex $u\in{U'}$  
(resp., $w\in{W'}$) in $\hatG^{bound}$. 
We may assume the worst case that all vertices removed have degree
$dD$ and arcs incident to them are not covered by $(f_1,f_2)$. 
So, $(f'_1,f'_2)$ still covers $\gamma'|E(\hatG^{avg})|$ arcs of
$\hatG^{bound}$, but the number of arcs in $\hatG^{bound}$ is smaller
than that of $\hatG^{avg}$. 
By the analysis of the construction size, with high probability,
$\hatG^{bound}$ has at least $(1-2^{-n/6})|E(\hatG^{avg})|$ arcs. 
Thus, $(f_1,f_2)$ covers at most 
$\frac{\gamma'}{1-2^{-n/6}} \leq 2\gamma' \leq 2\gamma$ arcs of
$\hatG^{bound}$.

The first inequality follows because $1-2^{n/6}\geq 1/2$ for large
enough $n$. 
Therefore, there is no labeling covering more than $2\gamma$ fraction
of arcs of $\hatG^{bound}$. 
Moreover, this happens with high probability. 

This completes the proof of Lemma~\ref{lmm:degree-reduction}.


\section{Getting Hardness in Terms of $\dem$}
\label{sec:hardness-demands}

In this section, we discuss how to obtain the hardness in terms of
demand pairs. 
We will give an example of the hardness of the rooted $k$-connectivity
problem on directed graphs.
The next theorem is implicit in the construction of 
Cheriyan~et~al.~\cite{CLNV12}.

\begin{theorem}[Implicit in \cite{CLNV12}] 
\label{thm:dir-rooted-props}
There is a polynomial-time approximation preserving reduction
such that, given an instance 
$(G=(U,W,E),\{\pi_e:e\in{E}\},L_1,L_2,c_1,c_2)$ of the minimum-cost
label~cover problem, outputs an instance $(\hatG,c,r,T)$ of the
rooted $k$-connectivity problem on directed graphs with
$k=O(\delta(G)$, where $\hatG$ is a directed graph, $c$ is a cost
function, $r$ is a root vertex and $T$ is a set of terminals.  
Moreover, the reduction has the following properties: 
\squishlist
\item Each terminal $t_{i,j}\in{T}$ corresponds to an arc
      $(u_i,w_j)\in E(G)$. 
\item The graph $\hatG$ can be partitioned into 
      $\hatG=\bigcup_{t_{i,j}\in{T}}\hatG_{i,j}$, where $E_{i,j}$
      is the union of all $r,t_{i,j}$-paths in $\hatG$. 
\item For any two partitions $\hatG_{i,j}$ and $\hatG_{i',j'}$ of 
      $\hatG$, where $i\neq i'$ and $j\neq j'$,
      there is a path from $\hatG_{i,j}$ to $\hatG_{i',j'}$ 
      (resp., from $\hatG_{i',j'}$ to $\hatG_{i,j}$) if only if 
      the label~cover graph $G$ has an arc $(u_i,w_{j'})$ 
      (resp., $(u_{i'},w_{j'})$). 
\squishend
\end{theorem}

The full discussions are provided in Appendix~\ref{sec:dir-rooted},
and the discussions for other problems are discussed in 
Appendix~\ref{sec:undir-rooted}, \ref{sec:vc-sndp}, and
\ref{sec:k-route-cut}.

Our goal is to reduce the number of terminals by merging some
terminals of the instance $(\hatG,c,r,T)$ of the rooted
$k$-connectivity problem on directed graphs as in
Theorem~\ref{thm:dir-rooted-props}. 
However, if we merge terminals $t_{i,j}$ and $t_{i',j'}$ such that
$\hatG_{i,j}$ and $\hatG_{i',j'}$ share some non-root vertex, then
this will cause us some problems.
For example, we might have some ``free path'' formed by concatenating
an $r,t_{i,j}$-path and an $r,t_{i',j'}$-path, or we might not have
enough openly disjoint paths to satisfy the connectivity requirement. 
Thus, we have to ensure that no two terminals that we merge share
a non-root vertex in the graphs $\hatG_{i,j}$'s.

Observe that if the label~cover graph $G$ has no arc joining $(u_i,w_j)$ and
$(u_{i'},w_{j'})$, where $i\neq{i'}$ and $j\neq{j'}$, then 
the graphs $\hatG_{i,j}$ and $\hatG_{i,j}$ share no non-root vertex.
In other words, if $(u_i,w_j)$ and $(u_{i'},w_{j'})$ form 
``an induced matching'' in $G$, then we can merge terminals $t_{i,j}$
and $t_{i',j}$. 
Hence, we can partition arcs of $G$ into induced matching by applying
``strong edge coloring'', and the number of partition of arcs we
obtain is at most $2\Delta(G)^2$. 
Thus, we can merge terminals in $T$ into $2\Delta(G)^2$ terminals.
Applying Theorem~\ref{thm:mincost-small-strong-lc}, we have the
hardness of $|T|^{1/4-\epsilon}=\dem^{1/4-\epsilon}$, for any
$\epsilon>0$ as claimed.


\medskip

\noindent{\bf Acknowledgment.} We thank Adrian Vetta, Joseph
Cheriyan, Guyslain Naves, Parinya Chalermsook, Danupon Nanongkai and 
Siu On Chan for useful comments and discussions.

\pagebreak

\bibliographystyle{plain}
\bibliography{hardness-conn}

\pagebreak

\appendix


\section{From Maximum To Minimum-Cost Label~Cover}
\label{sec:max-to-min}

In this section, we show how to obtain the hardness of the minimum-cost
label~cover problem from the hardness of the maximum label~cover
problem.

The following is a standard lemma that transforms the
hardness of the maximization version of the label~cover problem to the 
hardness of the minimum-cost label~cover problem. 
The theorem has been proved for the case that an instance has 
perfect completeness; see \cite{ABSS97,Kortsarz01}, and also see
\cite{CK07,CCK08}.
For our purpose, we state the theorem for the case of imperfect
completeness.

\begin{lemma}\label{lmm:max-to-min}
Suppose there are constants $0<\gamma,\epsilon<1$ and $q_1,q_2>0$ such
that, given an instance $(G=(U,W,E),\{\pi_e:e\in E\},L_1,L_2)$ of the  
maximum label cover problem with $\maxleft=q_1\avgleft$ and
$\maxright=q_2\avgright$, it is hard to distinguish between
the following two cases:
\squishlist
\item {\sc Completeness:} There is a labeling covering at least
  $1-\epsilon$ fraction of the arcs.
\item {\sc Soundness:} There is no labeling covering at least
  $\gamma$ fraction of the arcs. 
\squishend

Then it is hard to approximate the minimum-cost label~cover
problem to within a factor of $o(1/\sqrt{\gamma})$. 
\end{lemma}

\begin{proof}
We construct an instance of the minimum-cost label~cover problem
from an instance $(G=(U,W,E),\{\pi_e:e\in E\},L_1,L_2)$ of the maximum
label~cover problem with a parameter $\epsilon$ such that 
$\epsilon\cdot|E| \leq \min\{|U|,|W|\}$ as follows.  
First, we take an instance $(G=(U,W,E),\{\pi_e:e\in E\},L_1,L_2)$ as a
base construction.
Then we set costs $c_1$ and $c_2$ of the left and right labels so that
$c_1|U|=c_2|W|$, and let $C=c_1|U|+c_2|W|$.  
To show the hardness of the minimum-cost label~cover problem, it suffices
to show that there is a gap of at least
$\sqrt{\gamma}/(16\sqrt{2q_1q_2})=\Omega(\sqrt{\gamma})$ between the two  
cases of maximum label~cover instances.
\\

{\bf Completeness:} Suppose there is a labeling $(f_1,f_2)$
of the maximum label~cover instance that covers $1-\epsilon$ fraction
of the arcs. Then, clearly, there is a labeling
$(\hat{f}_1,\hat{f}_2)$ of the minimum-cost label~cover that covers the
same number of arcs.  For each arc $(u,w)$ not covered, we add to
$\hat{f}_1(u)$ and $\hat{f}_2(w)$ labels $a\in L_1$ and $b\in L_2$
such that $\pi_{u,w}(a)=b$. 
By the construction, the labeling $(\hat{f}_1,\hat{f}_2)$ covers all
the arcs.
Since $\epsilon|E| \leq \min\{|U|,|W|\}$, the cost of the
labeling $(\hat{f}_1,\hat{f}_2)$ is at most $2(c_1|U|+c_2|W|)=2C$.\\ 

{\bf Soundness:}
Suppose there is no labeling of the maximum 
label~cover instance that covers at least $\gamma|E|$ arcs.
We will show that if there is a labeling $(\hat{f}_1,\hat{f}_2)$
of the minimum cost label~cover instance with cost 
$\alpha C \leq (\sqrt{\gamma}/(8\sqrt{2q_1q_2}))\cdot C$, then there
is a labeling $(f_1,f_2)$ of the maximum label~cover instance that
covers at least $\gamma|E|$ arcs.

First, we construct $(f_1,f_2)$ from $(\hat{f}_1,\hat{f}_2)$ by
uniformly at random picking a label $a\in \hat{f}_1(u)$ and assigning
$f_1(u)=a$, for each $u\in U$, and uniformly at random picking a label
$b\in\hat{f}_2(w)$ and assigning $f_2(w)=b$, for each $w\in W$. 
We claim that $(f_1,f_2)$ covers at least $\gamma|E|$ arcs.
To see this, consider the number of labels assigned to
$(\hat{f}_1,\hat{f}_2)$. 
Let $U'\subseteq U$ and $W'\subseteq W$ be sets of vertices with at
most $8\alpha q_1$ and $8\alpha q_2$ labels, respectively, and let
$E'\subseteq E$ be the set of arcs with both endpoints in $U'\cup W'$.  
Then we have 
\[
|U\setminus U'| \leq \frac{\alpha C}{8\alpha\cdot q_1c_1}
                 =   \frac{1}{4q_1}|U| 
\mbox{ and }
|W\setminus W'| \leq \frac{\alpha C}{8\alpha\cdot q_2c_2} 
                 =   \frac{1}{4q_2}|W| 
\]

Thus, by union bound, the number of arcs of $E'$ is at least 
\[
|E|-\frac{1}{4q_1}\maxleft|U|-\frac{1}{4q_2}\maxright|W|
 \geq |E|-\frac{|E|}{2}
\]
The second inequality follows from the facts that
$\maxleft|U|=q_1\minleft|U|\leq q_1|E|$ and 
$\maxright|W|=q_2\maxright|W|\leq q_2|E|$.
The probability that $(f_1,f_2)$ covers any arc
$e\in E'$ is at least $1/(64\alpha^2)$. 
Thus, the expected number of arcs of $E'$ covered by
$(f_1,f_2)$ is 
\[
\sum_{e\in E'}(1\cdot \Pr[\mbox{$(f_1,f_2)$ covers $e$}]) 
\geq \frac{|E|}{2}\cdot\frac{1}{64\alpha^2\cdot q_1q_2}
\geq \frac{\gamma|E|}{128q_1q_2/(8\sqrt{2q_1q_2})^2} = \gamma |E|.  
\]
We can derandomize this process by the method of conditional
expectation. 
Therefore, there is a labeling $(f_1,f_2)$ of the maximum label~cover
instance that covers at least $\gamma|E|$ arcs, a contradiction.
\end{proof}

A similar lemma can be proven for the case of the label~cover problem 
with perfect completeness. We will skip the proof for this case
since it is almost identical to the previous one.

\begin{lemma}\label{lmm:max-to-min-perfect}
Suppose there are constants $0<\gamma<1$ and $q_1,q_2>0$ such
that, given an instance $(G=(U,W,E),\{\pi_e:e\in E\},L_1,L_2)$ of the  
maximum label cover problem with $\maxleft=q_1\minleft$ and
$\maxright=q_2\maxright$, it is hard to distinguish between
the following two cases:
\squishlist
\item {\sc Completeness:} There is a labeling covering all the arcs.
\item {\sc Soundness:} There is no labeling covering at least
  $\gamma$ fraction of the arcs. 
\squishend

Then it is hard to approximate the minimum-cost label~cover
problem to within a factor of $o(1/\sqrt{\gamma})$. 
\end{lemma}

\section{The Right Degree Reduction}
\label{sec:right-degree-reduction}

In this section, we discuss the right degree reduction introduced by
Moshkovitz and Raz in \cite{MR10}. 
The right~degree reduction is an operation that transforms any
instance of the maximum label~cover problem to an instance with
regular right~degree $d$ while preserving the completes and
preserving the soundness up to additive $O(1/\sqrt{d})$. 

The right~degree reduction is described as follows. 
Take an instance $(G=(U,W,E),\{\pi_e:e\in E\},L_1,L_2)$ of the maximum
label~cover problem with completeness $1-\epsilon$ and soundness
$\gamma$. 
For each right vertex $w\in W$, we construct an expander graph $H_w$ 
($H_w$ is an undirected graph) on $\deg(w)$ vertices with regular
degree $d$ and a second eigenvalue  $O(\sqrt{d})$. 
%
%
Then we replace each vertex $w\in W$ by vertices of $H_w$.
To be precise, we make $\deg(w)$ copies of $w$, namely 
$w(1),w(2),\ldots,w(\deg(w))$, and associate each vertex $w(j)$ to 
a vertex of $H_w$ by a one-to-one mapping. 
We order neighbors of $w$ in $G$ arbitrary, and let
$u_1,u_2,\ldots,u_{\deg(w)}$ be the neighbors of $w$. 
For each edge $\{w(i),w(j)\}$ of $H_w$, we add an arc $(u_i,w(j))$ and
place a constraint $\pi_{u_i,w_j}=\pi_{u_i,w_j}$ on the arc
$(u_i,w_j)$.

By the construction, there are $d$ copies of arcs $e\in E$ in the
output instance, and they have the same constraint.
Thus, for any labeling $(f_1,f_2)$ that covers $(1-\epsilon)$ fraction of
arcs of the input instance, there is a labeling $(f_1,f'_2)$ that
covers $(1-\epsilon)$ fraction of arcs of the output instance, where
$f'_2$ can be constructed by assigning $f'_2(w(j))=f_2(w)$ for all
copies $w(j)$ of a vertex $w\in{W}$.
It was shown in \cite{MR10} using the expander mixing lemma that the
right-degree reduction gives an output instance with soundness
$\gamma+O(1/\sqrt{d})$; see \cite{MR10} and \cite{DH09} for more
detail. 

This operation requires the projection property of a label~cover
instance and thus does not apply to the more general instance in which
the constraints $\pi_{e}$ are relations rather than projections.
Also, the additive loss $O(1/\sqrt{d})$ in the soundness is the best 
possible because the smallest possible second eigenvalue of the 
$d$-regular expander graph is $2\sqrt{d-1}$ due to the work of 
Alon and Boppana; see Theorem 5.3 in \cite{HLW06}.


\section{Rooted $k$-Connectivity  on Directed Graphs}  
\label{sec:dir-rooted}

In this section, we present hardness of $\Omega(k^{1/2})$ and
$\Omega(\dem^{1/4})$ for the rooted $k$-connectivity problem on
directed graphs.

\subsection{Hardness in Terms of $k$}
\label{sec:dir-term-k}

First, we give a hardness of $\Omega(k^{1/2})$ for the rooted
$k$-connectivity problem on directed graphs.
Our result is based on the construction in \cite{CLNV12}.
Here we will give a construction but will omit the proof. 
Let $(G=(U,W,E),\{\pi_e:e\in E\},L_1,L_2,c_1,c_2)$ be an instance of
the minimum-cost label~cover problem.
We construct a directed graph $\hatG=(\hatV,\hatE)$ of the
rooted $k$-connectivity problem as follows.

\medskip

\noindent {\bf Base Construction:} 
For each vertex $u_i\in U$, we add to $\hatG$ a vertex $u_i$ and a set
of vertices $A_i$, which is a copy of the set of labels $L_1$; we join
$u_i$ to each vertex $a\in A_i$ by an arc $(u_i,a)$. 
For each vertex $w_j\in W$, we add to $\hatG$ a vertex $w_j$ and a set
of vertices $B_j$, which is a copy of the set of labels $L_2$; we join
$w_j$ to each vertex $b\in B_j$ by an arc $(u_i,a)$. 
We may think that $u_i$ (resp., $w_j$) is the same vertex in both $G$
and $\hat G$.
Also, since $A_i$ (resp., $B_j$) is a copy of $L_1$ (resp., $L_2$), we
may say that a vertex $a\in A_i$ (resp., $b\in B_j)$ is a label in 
$L_1$ (resp., $L_2$).
We set cost $c_1$ on an arc $(u_i,a)$, for each $a\in A_i$, and we
set cost $c_2$ on an arc $(b,w_j)$, for each $b\in B_j$.
For each arc $(u_i,w_j)$ of $G$, we add to $\hatG$ a zero-cost arc
$(a,b)$ joining a vertex $a\in A_i$ to a vertex $b\in B_j$ if
$\pi_{u_i,w_j}(a)=b$. 
This finishes the base construction.

\medskip

\noindent {\bf The final construction:} 
Now, we add a root vertex $r$ to $\hatG$ and join $r$ to each vertex
$u_i$ by a zero-cost arc $(r,u_i)$. 
For each arc $(u_i,w_j)$, we add a terminal $t_{i,j}$ and join $w_j$
to $t_{i,j}$ by a zero-cost arc $(w_j,t_{i,j})$. 
Thus, we have the root vertex $r$ and a set of terminals 
$T_{i,j}=\{t_{i,j}:(u_i,w_j)\in E\}$. 
Next, we add a zero-cost arcs $(u_{i'},t_{i,j})$, called a 
{\em padding arc}, if $i'\neq i$ and $(u_{i'},w_j)\in E$. 
Thus, each terminal has indegree at most $\maxdeg(G)$. 
For each terminal $t_{i,j}$ with indegree $d_{i,j}<\maxdeg(G)$, we add
$\maxdeg(G)-d_{i,j}$ copies of a zero-cost arc $(r,t_{i,j})$. 
Finally, we set the connectivity requirement $k=\maxdeg(G)$. 
(Note that $\maxdeg(G) \ll |T|$.) 

\medskip

The above construction gives the following theorem whose correctness
is proved in \cite{CLNV12}.

\begin{theorem}[\cite{CLNV12}]\label{thm:dir-rooted}
There is a polynomial-time approximation preserving reduction
such that, given an instance of the minimum-cost label~cover problem
consisting of a graph $G$, outputs an instance of
the rooted $k$-connectivity problem on directed graphs with
$k=\maxdeg(G)$. 
\end{theorem}

Applying Theorem~\ref{thm:mincost-small-strong-lc}, it then
immediately follows that the hardness of the rooted $k$-connectivity
problem on directed graphs is $\Omega(k^{1/2-\epsilon})$, for any
$\epsilon>0$ (since $k=\maxdeg(G)$).
Thus, we have the next theorem.

\begin{theorem}\label{thm:dir-rooted-conn}
For $k<|T|$, unless $\NP=\ZPP$, it is hard to approximate the rooted   
$k$-connectivity problem on directed graphs to within a factor of
$o(k^{1/2})$.  
\end{theorem}

\subsection{Hardness in Terms of $\dem$}
\label{sec:dir-term-T}

Now, we show the hardness of the rooted $k$-connectivity problem on
directed graphs in terms of the other parameter. 
Specifically, we show a hardness of $\Omega(\dem^{1/4})$ for the  
rooted $k$-connectivity problem on directed graphs, where $\dem=|T|$
(since the demand pairs are between the root vertex $r$ and terminals 
in $T$).
We start from the previous construction and then merge some
terminals. 
The key idea is to merge terminals that do not share paths from the
root vertex.
To be precise, we say that two terminals $t_{i,j}$ and $t_{i',j'}$ are
{\em dependent} if there are an $r,t_{i,j}$-path $P$ and an
$r,t_{i',j'}$-path $P'$ that have a common vertex $v\neq r$; 
otherwise, we say that $t_{i,j}$ and $t_{i',j'}$ are {\em independent}.  
Observe that two terminals $t_{i,j}$ and $t_{i',j'}$ are dependent if
and only if arcs $(u_i,w_j)$ and $(u_{i'},w_{j'})$ of $G$ (of the
label~cover instance) are incident or there is an arc joining them.  
Specifically, $(u_i,w_j)$ and $(u_{i'},w_{j'})$ are
independent if and only if they form an induced matching in $G$. 
This proves in the lemma below. 
For notational convenience, we use $ij$ to mean an arc $(u_i,w_j)$ of 
$G$.

\begin{lemma} \label{lmm:ind-terms}
Any two terminals $t_{i,j}$ and $t_{i',j'}$ are independent if and only
if $ij$ and $i'j'$ forms an induced matching in $G$.
\end{lemma}

\begin{proof}

First, we prove the ``only if'' part.
Suppose $t_{i,j}$ and $t_{i',j'}$ are independent, but $\{ij,i'j'\}$
is not an induced matching in $G$.
Then we have two cases:~(1) $i=i'$ or $j=j'$ (2) $G$ has an arc $ij'$
or $i'j$. 
For the former case, every $r,t_{i,j}$-path and $r,t_{i',j'}$-path in
$G$ has to use either $u_i$ or $w_j$ and thus share a vertex.
This implies that $t_{i,j}$ and $t_{i',j'}$ are dependent, a
contradiction.
For the latter case, assume wlog that $G$ has an arc $ij'$. 
Then we must have added to $G$ a padding arc $u_i,t_{i',j'}$ by the
construction. Thus, there is an $r,t_{i',j'}$-path that shares a
vertex $u_i$ with an $r,t_{i,j}$-path, a contradiction.
This proves the ``only if'' part.

Next, we prove the ``if'' part. 
Suppose $ij$ and $i'j'$ form an induced matching in $G$.
Then $i\neq i'$ and $j\neq j'$. 
Also, $G$ has neither an arc $ij'$ nor an arc $i'j$. 
It then follows immediately by the construction that $\hatG$ has no
$r,t_{i,j}$-path and $r,t_{i',j'}$-path that share a common vertex.
This completes the proof.
\end{proof}

Lemma~\ref{lmm:ind-terms} allows us to apply a strong edge coloring
algorithm to the arcs of $G$, which are constraints of the label~cover
instance. 
It is known that every graph $G$ can be strongly colored using at most
$2\maxdeg(G)^2$ colors.
Since each color class forms an induced matching in $G$, no two of them 
are dependent. 
Thus, we can merge all the terminals corresponding to arcs of the same
color class into one terminal without any conflict. 
To be precise, for each color $C$, define 
$T_C=\{t_{i,j} : \mbox{$ij$  has color $C$}\}$.
Then we unify $T_C$ as a single terminal and set a connectivity
requirement $k|T_c|$ for this terminal. 
The new graph is denoted by $\hatG^{new}$.
Observe that any $k|T_C|$ openly disjoint $r,T_C$-paths in
$\hatG^{new}$ corresponds to $k$ openly disjoint $r,t_{i,j}$-paths for
every $t_{i,j}\in T_C$ in the original graph $\hatG$.  
Thus, there is a one-to-one mapping between the solution in the new
instance and that of the old instance, and both have the same hardness. 
To make a connectivity uniform, set $k^{new}=k\cdot\max_C|T_C|$ and
add $k^{new}-k|T|_C$ copies of a zero-cost arc $(r,T_C)$ for each
terminal $T_C$. 
By the construction, we have at most $2\maxdeg(G)^2$ terminals in the
new instance. 
Therefore, applying Theorem~\ref{thm:mincost-small-strong-lc}, we have
a hardness of $\Omega(\dem^{1/4-\epsilon})=\Omega(|T|^{1/4-\epsilon})$ 
for the rooted $k$-connectivity problem on directed graphs, for any
constant $\epsilon>0$.

\begin{theorem}\label{thm:dir-rooted-terms}
For $k\geq |T|$, unless $\NP=\ZPP$, it is hard to approximate the
rooted $k$-connectivity problem on directed graphs to within a factor
of $o(\dem^{1/4-\epsilon})=o(|T|^{1/4-\epsilon})$ for any constant
$\epsilon>0$. 
\end{theorem}

\section{The Rooted $k$-Connectivity Problem on Undirected
  Graphs.} 
\label{sec:undir-rooted}

In this section, we present hardness constructions for the rooted
$k$-connectivity on undirected graphs.

\subsection{Hardness in Terms of $k$}
\label{sec:undir-rooted:term-k}

Similar to the case of the directed graphs, the following theorem has
been proved in \cite{CLNV12}.

\begin{theorem}[\cite{CLNV12}]
\label{thm:undir-rooted:terms-k}
There is a polynomial-time approximation preserving reduction
such that, given an instance of the minimum-cost label~cover problem
consisting of a graph $G$ with a set of labels $L_1$ and $L_2$,
outputs an instance of the rooted $k$-connectivity problem on directed
graphs with $k=O(\maxdeg(G)^3\cdot\max\{|L_1|,|L_2|\} + \maxdeg(G)^5)$. 
\end{theorem}

We will present the hardness construction described in \cite{CLNV12} but
will skip the proof for completeness and soundness.
(For more detail, see \cite{CLNV12}.)
Take an instance $(G=(U,W,E),\{\pi_e:e\in E\},L_1,L_2,c_1,c_2)$ of the
minimum-cost label~cove problem.
We construct a graph $\hatG=(\hatV,\hatE)$ of the rooted
$k$-connectivity problem on undirected graphs as follows. 

\medskip

\noindent {\bf Base Construction:} The base construction is the
same as that in Section~\ref{sec:dir-term-k} except that we ignore the
direction of edges. 

\medskip

\noindent {\bf Add Root Vertex and Terminals:}
We add to $\hatG$ a root vertex $r$.
For each arcs $(u_i,w_j)\in G$, we add to $\hatG$ a clique $X_{i,j}$
and a terminal $t_{i,j}$; the size of $X_{i,j}$ will be specified
later. 
We join each clique $X_{i,j}$ to a vertex $u_i\in U$ by adding a
zero-cost edge $\{x,u_i\}$ for each pair of vertices $x\in X_{i,j}$
and $u_i\in U$.
We join each terminal $t_{i,j}$ to a vertex $w_j\in W$ by a zero-cost
edge $\{w_j,t_{i,j}\}$. 
Then we join the root vertex $r$ to each clique $X_{i,j}$ by a
zero-cost edge $\{r,x\}$ for each $x\in X_{i,j}$. 

\medskip

\noindent {\bf Final Construction:} 
Now, we add some zero-cost edges, called {\em padding edges}, which
intuitively force $r,t_{i,j}$-paths to be in a {\em canonical} form.  
We say that an $r,t_{i,j}$-paths is a {\em canonical path} if it is of
the form $r,X_{i,j},u_i,A_i,B_j,w_j,t_{i,j}$. 
The padding for each terminal $t_{i,j}$ is as follows.
For notational convenience, we use $ij$ to means an arc $\{u_i,w_j\}$
of the label~cover instance $G$, and we use $\dist(ij,i'j')$ to mean
the distance between $ij$ and $i'j'$ in the {\em line graph} $H$ of
the underlying undirected of $G$, i.e., the vertex set of $H$ is the 
edge set of $G$, and there is an edge $(e,e')$ in $H$ if edges $e$
and $e'$ share an endpoint in $G$. 
We define the set $Z_{i,j}$ and $Y_{i,j}$ as below. 

\begin{align*}
Z^{1}_{i,j} &= 
  \left(\bigcup_{i'\neq i:i'j\in E}A_{i'}\right)\cup
  \left(\bigcup_{j'\neq j:ij'\in E}B_{j'}\right)\\
Z^{2}_{i,j} &= 
  \{t_{i',j'}:1\leq \dist(ij,i'j')\leq 2\}\\
Z_{i,j} &= Z^{1}_{i,j}\cup Z^{2}_{i,j}\\
Y_{i,j} &= \bigcup_{1\leq \dist(ij,i'j')\leq 2}X_{i',j'}\\
\end{align*}  

We also create a set of vertices $Q_{i,j}$, which is the set of
auxiliary vertices created to make the connectivity requirement
uniform. The vertices of $Q_{i,j}$ are not in the base construction,
and its size will be specified later.
We join $Y_{i,j}$ to $t_{i,j}$ by edges
$\{y,t_{i,j}\}$ for all $y\in Y_{i,j}$.
We join $Z_{i,j}$ to $X_{i,j}$ and $t_{i,j}$ by edges
$\{x,z\}$,$\{z,t_{i,j}\}$ for all pairs of vertices $x\in X_{i,j}$ and
$z\in Z_{i,j}$. 
We join $r,Q_{i,j}$ and $t_{i,j}$ by edges $\{r,q\}$,
$\{q,t_{i,j}\}$ for all $q\in Q_{i,j}$. 
All of these edges have zero costs.

Lastly, we have to set the connectivity requirement and
specify the size of $X_{i,j}$.
We remark that we want all the $r,t_{i,j}$-paths except a canonical
path to use all the vertices in $Z_{i,j}\cup Y_{i,j}\cup Q_{i,j}$. 
Thus, we need to set the size of $X_{i,j}$ to be 
$|Z_{i,j}|+1$. 
By the construction, we have to set
$|Q_{i,j}|=k-|Z_{i,j}|+|Y_{i,j}|-1$.
\[
 |Z_{i,j}|=|Z^{1}_{i,j}|+|Z^{2}_{i,j}|
 \leq 2\maxdeg(G)\cdot\max\{|L_1|,|L_2|\} + 2\maxdeg(G)^2. 
\]
Thus, we need 
$|X_{i,j}|=2\maxdeg(G)\cdot\max\{|L_1|,|L_2|\} + 2\maxdeg(G)^2+1$,
implying that 
$|Y_{i,j}|=O(\maxdeg(G)^3\cdot\max\{|L_1|,|L_2|\} + \maxdeg(G)^4)$. 
We set $k=\max_{i,j}(|Z_{i,j}|+|Y_{i,j}|)+1$ and set the size of
$Q_{i,j}$ to be $|Q_{i,j}|=k-|Z_{i,j}|+|Y_{i,j}|-1$.
Therefore, we have an instance of the rooted $k$-connectivity
problem on undirected graphs with 
$k=O(\maxdeg(G)^3\cdot\max\{|L_1|,|L_2|\} + \maxdeg(G)^4)$
(given an instance in Theorem~\ref{thm:mincost-small-strong-lc}, we 
have $k=O(q^5)$), and the reduction preserves both the
completeness and soundness. 
Here we skip the completeness and soundness proofs. 
For more detail, please see \cite{CLNV12}.

\subsection{Hardness in terms of $\dem$}
\label{sec:undir-rooted:terms-T}

To obtain the hardness in terms of $\dem$ for the rooted
$k$-connectivity problem on undirected graphs, we may apply the same
technique as that used in the directed case.
However, we can simplify the proof by applying the following theorem
due to Lando and Nutov~\cite{LN09}.

\begin{theorem}[\cite{LN09}]
\label{thm:dir-undir}
There is a polynomial-time approximation preserving reduction that,
given an instance of the ``{\em directed}'' rooted $k$-connectivity
problem consisting of a directed graph $G$ on $n$ vertices, a root
vertex $r$, a set of terminals $T$ and a connectivity requirement $k$, 
outputs an instance of the ``{\em undirected}'' rooted
$k$-connectivity problem consisting of an undirected graph $G'$ on
$n'=2n$ vertices, a root vertex $r$, and a set of terminals $T'$,
where $|T'|=|T|$ and $k'=k+n$. 
\end{theorem}

Since there is an $\Omega(|T|^{1/4-\epsilon})$-hardness for the
rooted $k$-connectivity problem on directed graphs, the same hardness
applies for the undirected case as well.

\begin{theorem}\label{thm:undir-rooted:terms-T}
For $k\geq |T|$, it is NP-hard to approximate the undirected
rooted  $k$-connectivity problem to within a factor of
$o(|T|^{1/4-\epsilon})=o(\dem^{1/4-\epsilon})$ for any constant $\epsilon>0$.
\end{theorem}

\section{Vertex-Connectivity Survivable Network Design}
\label{sec:vc-sndp}

In this section, we present hardness constructions of the
vertex-connectivity survivable network design problem on undirected
graphs. 

\subsection{Hardness in Terms of $k$}
\label{sec:vc-sndp:terms-k}

The hardness of the vertex-connectivity survivable network design
problem can be derived from its special case, the rooted
$k$-connectivity problem. However, by applying the reduction directly
from the minimum-cost label~cover problem, we have a better bound.

The following theorem is proved by Chakrabarty, Chuzhoy and Khanna in
\cite{CCK08}. 
\begin{theorem}[\cite{CCK08}]
\label{thm:vc-sndp}
There is a polynomial-time approximation preserving reduction
such that, given an instance of the minimum-cost label~cover problem
consisting of a graph $G$ with a set of labels $L_1$ and $L_2$,
outputs an instance of the rooted $k$-connectivity problem on directed
graphs with $k=O(\maxdeg(G)\cdot\max\{|L_1|,|L_2|\}+\maxdeg(G)^2)$. 
\end{theorem}

We will present the hardness construction described in \cite{CCK08}
but will skip the proof for completeness and soundness.
(For more detail, see \cite{CCK08}.)
Take an instance $(G=(U,W,E),\{\pi_e:e\in E\},L_1,L_2,c_1,c_2)$ of the
minimum-cost label~cove problem.
We construct a graph $\hatG=(\hatV,\hatE)$ of the vertex-connectivity
survivable network design problem on undirected graphs as follows.

\medskip

\noindent {\bf Base Construction:} The base construction is the
same as that in Section~\ref{sec:dir-term-k} except that we ignore the
direction of edges. 

\medskip

\noindent {\bf Add Source-Sink Pairs:}
We will add to the undirected graph $\hatG$ {\em source-sink} pairs,
i.e., we add a pair of vertices whose connectivity requirement is
positive.  
For each arcs $(u_i,w_j)\in G$, we add to $\hatG$ a source $s_{i,j}$
and a sink $t_{i,j}$. 
We join each source $s_{i,j}$ to a vertex $u_i\in U$ by adding a
zero-cost edge $\{s_{i,j},u_i\}$, and we join each terminal $t_{i,j}$
to a vertex $w_j\in W$ by a zero-cost edge $\{w_j,t_{i,j}\}$. 

\medskip

\noindent {\bf Final Construction:} 
Now, we add some zero-cost edges, called {\em padding edges}, which
intuitively force $s_{i,j},t_{i,j}$-paths to be in a {\em canonical}
form. We say that an $r,t_{i,j}$-paths is a {\em canonical path} if it
is of the form $s_{i,j},u_i,A_i,B_j,w_j,t_{i,j}$. 
The padding for a source-sink pair $s_{i,j},t_{i,j}$ is as follows.
For notational convenience, we use $ij$ to means an arc $\{u_i,w_j\}$
of the label~cover instance $G$, and we use $\dist(ij,i'j')$ to mean
the distance between $ij$ and $i'j'$ in the line graph $H$ of 
the underlying undirected of $G$, i.e., the vertex set of $H$ is the 
edge set of $G$, and there is an edge $(e,e')$ in $H$ if edges $e$
and $e'$ share an endpoint in $G$. 
We define the set $Z_{i,j}$ and $Y_{i,j}$ as below. 

\begin{align*}
Z_{i,j} &= 
  \left(\bigcup_{i'\neq i:i'j\in E}A_{i'}\right)\cup
  \left(\bigcup_{j'\neq j:ij'\in E}B_{j'}\right)\\
Y_{i,j} &= \bigcup_{1\leq \dist(ij,i'j')\leq 2}\{s_{i',j'},t_{i',j'}\}\\
\end{align*}  

We join $s_{i,j}$ and $t_{i,j}$ to $Y_{i,j}\cup Z_{i,j}$ by adding
zero-cost edges $\{s_{i,j},x\}$ and $\{x,t_{i,j}\}$ for all
$x\in{Y_{i,j}\cup{Z_{i,j}}}$.
For the connectivity requirement, we set 
$\req(s_{i,j},t_{i,j})=|Y_{i,j}\cup{Z_{i,j}}|+1$ for all source-sink
pairs $s_{i,j},t_{i,j}$. 
We may make the requirements uniform by setting
$k=\max_{i,j}|Y_{i,j}\cup{Z_{i,j}}|+1$ and adding a set of auxiliary 
vertices $Q_{i,j}$ with $|Q_{i,j}|=k-|Y_{i,j}\cup{Z_{i,j}}|-1$ for
each source-sink pair $s_i,t_j$.

By the construction, we have 
\[
 |Z_{i,j}| + |Y_{i,j}| 
  \leq 2\maxdeg(G)\cdot\max\{|L_1|,|L_2|\} + 4\maxdeg(G)^2
\]
Therefore, we have an instance of the rooted $k$-connectivity
problem on undirected graphs with 
$k=O(\maxdeg(G)\cdot\max\{|L_1|,|L_2|\} + \maxdeg(G)^2)$
(given an instance in Theorem~\ref{thm:mincost-small-strong-lc}, we 
have $k=O(q^3)$), and the reduction preserves both the 
completeness and soundness. 
Here we skip the completeness and soundness proofs. 
For more detail, please see \cite{CCK08}.

\subsection{Hardness in Terms of $\dem$}
\label{sec:vc-sndp:terms-T}

The hardness in terms of $\dem$ for the vertex-connectivity
survivable network design problem on undirected graphs follows
immediately from that of its special case, the rooted $k$-connectivity
problem on undirected graphs.
Thus, we have 

\begin{theorem}\label{thm:vc-sndp:terms-T}
For $k\geq\dem$, it is NP-hard to approximate the vertex-connectivity
survivable network design problem on undirected graphs to within a
factor of $o(\dem^{1/4-\epsilon})$ for any constant $\epsilon>0$. 
\end{theorem}

\section{Vertex-Connectivity $k$-Route Cut}
\label{sec:k-route-cut}

In this section, we discuss the vertex-connectivity $k$-route cut
problem. 

\subsection{Hardness in Terms of $k$}
\label{sec:k-route-terms-k}

The following theorem is proved in \cite{CMVZ12} by Chuzhoy~et~al. 

\begin{theorem}[\cite{CMVZ12}]
\label{thm:vc-k-route}
There is a polynomial-time approximation preserving reduction
such that, given an instance of the minimum-cost label~cover problem
consisting of a graph $G$ with a set of labels $L_1$ and $L_2$,
outputs an instance of the vertex-connectivity $k$-route cut problem
on undirected graphs with $k=O(\maxdeg(G)\cdot\max\{|L_1|,|L_2|\})$.
\end{theorem}

We will give the hardness construction of this problem based on the
construction in \cite{CMVZ12}\footnote{
Due to a very subtle error in the proof in \cite{CMVZ12}, our
construction is slightly different from the original construction.
}.
Take an instance $(G=(U,W,E),\{\pi_e:e\in E\},L_1,L_2,c_1,c_2)$ of the
minimum-cost label~cove problem.
We construct a graph $\hatG=(\hatV,\hatE)$ of the vertex-connectivity
$k$-route cut problem as follows.

\medskip

\noindent {\bf Base Construction:} 
First, for each left vertex $u_i\in U$, we create a set of edges
$E(u_i)=\{\{a,a'\}: a \in L_1\}$. Similarly, 
for each right vertex $w_j\in W$, we create a set of edges 
$E(w_j)=\{\{b,b'\}: b \in L_2\}$. 
All the edges in $E(u_i)$'s have costs $c_1$, and all the edges in 
$E(w_j)$'s have costs $c_2$.

Next, for each right vertex $w_j\in W$, we arrange edges in $E(w_j)$
in an arbitrary order, say
$E(w_j)=\{\{b_1,b'_1\},\{b_2,b'_2\},\ldots,\{b_{|L_2|},b'_{|L_2|}\}\}$.
We then form a path $P_j$ by joining edges in $E(w_j)$ by
edges with costs infinity. 
To be precise, we have an edge $\{b_{\ell},b_{\ell+1}\}$ with cost
infinity in $\hatG$, for each $\ell=1,2,\ldots,|L_2|-1$. 
Because of the projection property, which we will discuss later, we
only create such paths for right vertices.

For each arc $(u_i,w_j)\in E$, we construct a path $Q_{i,j}$ as follows. 
For any label $b_\ell\in L_2$,  
let $\pi^{-1}(b_\ell)=(a_{q_1},a_{q_2},\ldots,a_{q_{\alpha(b_\ell)}})$
be a sequence of labels in $L_1$ that projects to $b_{\ell}$ arranged
in an arbitrary order. 
Since each label $a_{q_t}$ maps to an edge $\{a_{q_t},a'_{q_t}\}$ in
$E(w_j)$, we may abuse $\pi^{-1}(b)$ to mean edges in $E(u_i)$. 
We define an order $\psi(i,j)$ by arranging the edges of $E(u_i)$
to be of the form
\[
\pi^{-1}(b_1,)\pi^{-1}(b_2),\ldots,\pi^{-1}(b_{|L_2|})
\]
Then we form a path $Q_{i,j}$ by joining edges in $E(u_i)$
according to the order $\psi(i,j)$ by paths of length $2$.
That is,
\[
 Q_{i,j}=(a_1,a'_1,x_1,a_2,a'_2,x_2,\ldots,a_{|L_1|},a'_{|L_1|})
\]
where indices of $a_q$'s are obtained from $\psi(i,j)$. 
(This is crucial as edges in $E(u_i)$ may have different orders in
two different paths $Q_{i,j}$ and $Q_{i,j'}$.)
We denote by $X_{i,j}$ a set of vertices $x_q's$ in $Q_{i,j}$ separating
edges of $E(u_i)$. 
Next, for each edge $\{b_{\ell},b'_{\ell}\}$ and
$\{b_{\ell+1},b'_{\ell+1}\}$ in $E(w_j)$, we join  $b'_{\ell}$ and
$b_{\ell+1}$ to a vertex $x_{q_\ell}$, where $x_{q_{\ell}}$ is a
vertex in $X_{i,j}$ connecting the path on $\pi^{-1}(b_{\ell})$ to the
path on $\pi^{-1}(b_{\ell+1})$. These edges have costs infinity.
This completes the base construction. 

\medskip

\noindent {\bf Add Source-Sink Pairs:} 
For each vertex $u_i\in U$, we add a source vertex $s_i$, and for each
vertex $w_j\in W$, we add a sink vertex $t_j$. 
For each arc $(u_i,w_j)$ of $G$, we add a demand (source-sink) pair
$(s_i,t_j)$ to the set of demand pairs $\mathbb{D}$; then we join $s_i$
to the first vertices of the paths $Q_{i,j}$ and $P_j$, and we
join $t_j$ to the last vertices of the paths $Q_{i,j}$ and
$P_j$. All of these edges have costs infinity.
\medskip

\noindent {\bf Final Construction:} 
Now, we add {\em padding edges}, which will guarantee that to make the
vertex-connectivity of a source-sink pair $s_i,t_J$ to be below $k$,
ones have to remove edges corresponding to the feasible labeling of
the label~cover instance.
For the ease of presentation, we will use $ij$ to mean an arc
$(u_i,w_j)$ of $G$. 
Denote by $V(F)$ a set of vertices spanned by a set of edges (resp.,
a graph) $F$. 
We define a set of vertices $Z_{i,j}$ to be the set of neighbors of
$V(Q_{i,j})\cup{V(P_{j})}$ in the current graph; that is, 
\[
Z_{i,j} = \bigcup_{i'j\in E(G):i'\neq i}(V(Q_{i',j})\cup\{s_{i'}\}) \cup
         \bigcup_{ij'\in E(G):j'\neq j}
             (X_{i,j'}\cup{V(P_{j'})\cup\{t_{j'}\}})
\]
By the construction, we have
\[
z = \max_{ij\in E}|Z_{i,j}| \leq O(\maxdeg(G)(|L_1|+|L_2|)
\leq O(\maxdeg(G)\cdot\max\{|L_1|,|L_2|\}).
\]
The size $Z_{i,j}$ and $Z_{i',j'}$ for $ij\neq i'j'$ may be
different.
Thus, we add to $\hatG$ a set of auxiliary vertices $S_{i,j}$ with 
$|S_{i,j}|=z-|Z_{i,j}|$.
Then we add edges $\{s_i,v\},\{t_j,v\}$ with cost infinity joining $s_i$
and $t_j$ to each vertex $v\in Z_{i,j}\cup{S_{i,j}}$. 
Note that all the neighbors of 
$\{s_i,t_j\}\cup{V(Q_{i,j})\cup{V(P_{j})}}$ are in
$Z_{i,j}\cup{S_{i,j}}$.  
Finally, we set $k=z+1$, finishing the construction.

\medskip

\noindent{\bf Completeness and Soundness:} 
Both the completeness and soundness proofs follows from the next claim.

\begin{claim}\label{clm:remove-padding}
Consider any source-sink pair $(s_i,t_j)\in\mathbb{D}$.
There are at most $k+1$ openly disjoint $s_i,t_j$-paths in $\hatG$,
and there are at most $2$ openly disjoint $s_i,t_j$-paths in
$\hatG\setminus(Z_{i,j}\cup S_{i,j})$.
Moreover, any $2$ openly disjoint $s_i,t_j$-paths in
$\hatG\setminus(Z_{i,j}\cup S_{i,j})$ must have $Q_{i,j}$ and $P_j$ as
subpaths. 
\end{claim}
\begin{proof}
First, observe that vertices in $Z_{i,j}\cup{S_{i,j}}$, $V(Q_{i,j})$
and $V(P_j)$ are pairwise disjoint by the construction.  
It can be seen that $Z_{i,j}\cup{S_{i,j}}$ gives $k-1$ paths between
$s_i$ and $t_j$, and the path $Q_{i,j}$ and $P_j$ give another two
paths.
Thus, we proved both the first and the second statements.
For the third statement, it follows by the construction 
that all the neighbors of $\{s_i,t_j\}\cup{V(Q_{i,j})}\cup{V(P_j)}$
are in $Z_{i,j}\cup{S_{i,j}}$, which means that the only way we can
have $2$ openly disjoint $s_i,t_j$-paths in
$\hatG\setminus(Z_{i,j}\cup S_{i,j})$ is to follow the two paths
$Q_{i,j}$ and $P_j$. Hence, the claim follows.
\end{proof}

It can be seen that there is a one-to-one mapping between edges with 
finite cost and the labels of the minimum-cost label~cover instance.
Thus, its suffices to show that, for any edges $E'\subseteq E$ with
finite cost, $G\setminus{E'}$ has no $k$ openly disjoint disjoint
$s_i,t_j$ paths for all $(s_i,t_j)\in\mathbb{D}$ if and only if the
labeling $(f_1,f_2)$ corresponding to $E'$ is feasible for the
minimum-cost label~cover instance.
To see this, consider the graph
$\hatG_{i,j}=\hatG\setminus(Z_{i,j}\cup{S_{i,j}})$ for any source-sink
pair $(s_i,t_j)\in\mathbb{D}$. 
By Claim~\ref{clm:remove-padding}, there are at most $2$ openly
disjoint $s_i,t_j$-paths, and these paths have to use $Q_{i,j}$ and
$P_j$. 
Now, consider the paths $Q_{i,j}$ and $P_j$.
Observe that if we remove an edge $\{b,b'\}$ and one edge
in $\pi^{-1}_{u_i,w_j}$, then the resulting graph has no
$s_i,t_j$-path.
Conversely, if we remove $\{b,b'\}$ but none of the edges
in $\pi^{-1}_{u_i,w_j}$, then we can have a path that goes zig-zag
between $Q_{i,j}$ and $P_j$ via an edge $\{x_{q},\hat{b}\}$, where
$\{\hat{b},\hat{b}'\}$ is an edge next to $\{b,b'\}$. 
Therefore, we conclude that there is no $k$ openly disjoint
$s_i,t_j$-paths in $\hatG_{i,j}$ if and only if we remove edges
corresponding to the labeling covering an edge $\{u_i,w_j\}$ of $G$. 
This completes the proof of Theorem~\ref{thm:vc-k-route}.

\qed

\subsection{Hardness in Terms of $\dem$}
\label{sec:vc-k-route:terms-k}

For the hardness in terms of the number of demand pairs, it can be
seen that for any pair of arcs $(u_i,w_j)$ and $(u_{i'},w_{j'})$ that
form an induced matching (i.e., $G$ has no arc joining $(u_i,w_j)$
and $(u_{i'},w_{j'})$), all the vertices and edges of subgraphs of 
$\hatG$ induced by  
$Z_{i,j}\cup{S_{i,j}}\cup{V(Q_{i,j})}\cup{V(P_j)}$ and 
$Z_{i',j'}\cup{S_{i',j'}}\cup{V(Q_{i',j'})}\cup{V(P_{j'})}$ are
disjoint by the construction.
Thus, two arcs $(u_i,w_j)$ and $(u_{i'},w_{j'})$ are {\em independent}
if and only if $(u_i,w_j)$ and $(u_{i'},w_{j'})$ form an induced
matching. 
By the same arguments as in the case of the rooted $k$-connectivity
problem on directed graphs, we have the following theorem.

\begin{theorem}\label{thm:vc-k-route-terms-D}
For $k\geq\dem$, unless $\NP=\ZPP$, it is hard to approximate the
vertex connectivity $k$-route cut problem to within a factor
of $o(\dem^{1/4-\epsilon})$ for any constant $\epsilon>0$.   
\end{theorem}

\end{document}